\documentclass[runningheads]{llncs}
\usepackage{amssymb}
\usepackage{amsmath,bm}
\usepackage{graphicx}
\usepackage{multirow}
\usepackage{array}
\usepackage{algorithm}
\usepackage{algorithmic}
\usepackage{subfig}
\usepackage{float}
\usepackage{diagbox}

\newcommand{\PreserveBackslash}[1]{\let\temp=\\#1\let\\=\temp}

\newcolumntype{C}[1]{>{\PreserveBackslash\centering}p{#1}}
\begin{document}
\title{The Complexity of the Partition Coloring Problem
}
%
%
\author{Zhenyu Guo\inst{1} \and
Mingyu Xiao\inst{2} \and
Yi Zhou\inst{3}}
%
%
\institute{University of Electronic Science and Technology of China \\
\email{Harry.Guo@outlook.com} \and
University of Electronic Science and Technology of China\\
\email{myxiao@gmail.com} \and
University of Electronic Science and Technology of China\\
\email{zhou.yi@uestc.edu.cn}}
\maketitle              
\begin{abstract}
    Given a simple undirected graph $G=(V,E)$ and a partition of the vertex set $V$ into $p$ parts, the \textsc{Partition Coloring Problem} asks if we can select one vertex from each part of the partition such that the chromatic number of the subgraph induced on the $p$ selected vertices is bounded by $k$.
  PCP is a generalized problem of the classical \textsc{Vertex Coloring Problem} and has applications in many areas, such as scheduling and encoding etc.
  In this paper, we show the complexity status of the \textsc{Partition Coloring Problem} with three parameters:
  the number of colors, the number of parts of the partition, and the maximum size of each part of the partition.
  Furthermore, we give a new exact algorithm for this problem.
    \end{abstract}

    \section{Introduction}
    Given a simple undirected graph $G=(V,E)$, the vertex coloring is to assign each vertex a color such that no two adjacent vertices have the same color. In the \textsc{Vertex Coloring Problem} (VCP), a graph $G$ together with an integer $k$ are given, and the goal is to decide whether $G$ can be colored by using at most $k$ colors~\cite{karp1972reducibility}. VCP is an important problem in both graph theory and practice~\cite{galinier1999hybrid,gamst1986some,glass2003genetic}.

  In this paper, we study a generalized version of VCP, the \textsc{Partition Coloring Problem (PCP)},  which is also called the \textsc{Selective Graph Coloring Problem} in
  some references~\cite{demange2015some}.
  In PCP, we are given a graph $G=(V,E)$ with a partition $\mathcal{V}$ of the vertex set and an integer $k$, where
  $\mathcal{V}=\{V_1, V_2, \cdots, V_p\}$, $V_i\cap V_j = \emptyset$  for all $1\le i,j\le p$, and $\bigcup_{1\le i\le p} V_i = V$. The problem asks whether
  there is an induced subgraph containing exactly one vertex from each part $V_i$ of the partition $\mathcal{V}$ that is colorable by using $k$ colors.
  Note that when each part $V_i$ ($1\le i \le p$) is a singleton, i.e., $|V_i|=1$, PCP is equal to VCP. So VCP is a special case of PCP.
  Indeed, PCP, together with other extended coloring problems, such as the \textsc{Coloring Sum Problem}~\cite{kubicka1989introduction}, the \textsc{Edge Coloring Problem}~\cite{holyer1981np}, the \textsc{Mixed Graph Coloring Problem}~\cite{hansen1997mixed}, the \textsc{Split Coloring Problem}~\cite{ekim2005split} and so on, have been intensively studied from the view of computational complexity in the last decades~\cite{jin2017algorithms,zhou2004algorithm,Damaschke2019,lucarelli2010max}.

  \subsection{Existing literature}
  The literature dealing with PCP is rich and diverse. In terms of applications, PCP was firstly introduced in~\cite{li2000partition} to solve the wavelength routine and assignment problem, which is to assign a limited number of bandwidths on fiber networks. PCP also finds applications in a wide range in dichotomy-based constraint encoding, antenna positioning and frequency assignment and scheduling~\cite{demange2015some}.

  PCP is NP-Complete in general since the well-known NP-Complete problem, VCP is a special case of PCP. In pursuit of fast solution methods for PCP, heuristic searches without guarantee of the optimality represent one of the most popular approaches \cite{li2000partition,pop2013memetic,furini2018exact}. Exact algorithms based on integer linear programming were also investigated for solving problems of small scale~\cite{furini2018exact,frota2010branch,hoshino2011branch}.
  In terms of computational complexity, the NP-Hardness of this problem on special graph classes, such as  paths, circles, bipartite graphs, threshold graphs and split graphs are
  studied~\cite{demange2014complexity,demange2015some}.

  \subsection{Our contributions}
  In this paper, we further study the computational complexity of PCP.
  We always use $k$ to denote the number of colors, $p$ to denote the number of parts in the partition $\mathcal{V}$, and $q$ to denote the upper bound of the size  of all parts in the partition $\mathcal{V}$.
  We give some boundaries between P and NPC for this problem with different constant settings on the three parameters.
  We also consider the parameterized complexity of PCP: we show that PCP parameterized by $p$ is W[1]-hard and PCP parameterized by both $p$ and $q$ is FPT. The main complexity results  are summarized in Table~\ref{t1} and Table~\ref{t2}. In addition, we give a fast exact algorithm for PCP, which is based on subset convolution and runs in $O((\frac{n+p}{p})^pn\log k)$ time. Note that when $p=n$, PCP becomes VCP and the running time bound becomes $O(2^nn\log k)$, the best known running time bound for VCP.

  \begin{table}[]
   \caption{Complexity results with different constants $q$ and $k$}\label{t1}
   \centering
  \begin{tabular}{|c|c|c|c|}
  \hline
  \diagbox{$q$}{$k$}        & $k=1$                                                     & $k=2$                                                   & $k\ge 3$                                                  \\ \hline
  $q=1$    & \begin{tabular}[c]{@{}c@{}}P\\ (Theorem 2)\end{tabular}   & \begin{tabular}[c]{@{}c@{}}P\\ (Theorem 2)\end{tabular} & \begin{tabular}[c]{@{}c@{}}NPC\\ (Theorem 2)\end{tabular} \\ \hline
  $q=2$    & \begin{tabular}[c]{@{}c@{}}P\\ (Theorem 3)\end{tabular}   & \multicolumn{2}{c|}{\multirow{2}{*}{\begin{tabular}[c]{@{}c@{}}NPC\\ (Corollary 1)\end{tabular}}}                   \\ \cline{1-2}
  $q\ge 3$ & \begin{tabular}[c]{@{}c@{}}NPC\\ (Theorem 5)\end{tabular} & \multicolumn{2}{c|}{}                                                                                               \\ \hline
  \end{tabular}
  \end{table}

  \begin{table}[]
  \caption{Complexity results with parameters $p$ and $q$}\label{t2}
  \centering
  \begin{tabular}{|c|c|}
  \hline
  $p$ is constant              & \begin{tabular}[c]{@{}c@{}}P\\ (Theorem \ref{thm_pconstant})\end{tabular} \\ \hline
  Parameterized by $p$         & \begin{tabular}[c]{@{}c@{}}W{[}1{]}-hard\\ (Theorem \ref{p_w1})\end{tabular}        \\ \hline
  Parameterized by $p$ and $q$ & \begin{tabular}[c]{@{}c@{}}FPT\\ (Theorem \ref{p_FPT})\end{tabular}                  \\ \hline
  \end{tabular}
  \end{table}

  \section{Preliminaries}
  Let $G=(V,E)$ stand a simple and undirected graph with $n=|V|$ vertices and $m=|E|$ edges.
  For a vertex subset $X\subseteq V$, we use $G[X]$ to denote the subgraph induced by $X$.
  For a vertex $v\in V$, the set of vertices adjacent to $v$ is called the set of \emph{neighbors} of $v$ and denoted by $N(v)$.
  A graph is called a \emph{clique} if there is an edge between any pair of vertices in the graph and a graph is called an \emph{independent set} if there is no edge between any pair of vertices.

  For a nonnegative integer $k$, a $k$-coloring in a graph $G=(V,E)$ is a function $c:V\rightarrow \{1,2,\cdots,k\}$ such that
  for any edge $vu$ it holds that $c(v)\neq c(u)$. A graph is \emph{$k$-colorable} if it allows a $k$-coloring.
  The \textsc{Vertex Coloring Problem (VCP)} is to determine whether a given graph is  $k$-colorable.
  The smallest integer $k$ to make $G$ $k$-colorable is called the \emph{chromatic number} of $G$ and denoted by $\chi(G)$.

  Given an integer $p$, a \emph{$p$-partition} of the vertex set $V$ of $G$ is denoted by $\mathcal{V} = \{V_1, V_2, \cdots, V_p\}$, where $V_i\cap V_j = \emptyset$ for any pair of different $i$ and $j$ in $\{1,2,\dots,p\}$ and $\bigcup_{1\le i\le p} V_i = V$.
  Each subset of a $p$-partition is also called a \emph{part}. The maximum size of the parts in a $p$-partition
  $\mathcal{V} = \{V_1, V_2, \cdots, V_p\}$ is denoted by $q$, i.e., $q = \max_{i=1}^p |V_i|$.
  A \emph{selection} of a $p$-partition $\mathcal{V}$ 
  is a subset of vertex $S\subseteq V$ such that $|S\cap V_i| = 1$ for any $i\in \{1,2,\dots,p\}$.
  The \textsc{Partition Coloring Problem} is formally defined as follows.

  \noindent\rule{\linewidth}{0.3mm}
  \textbf{The Partition Coloring Problem (PCP)}\\
  \textbf{Input:} a graph $G=(V,E)$, a $p$-partition $\mathcal{V} = \{V_1, V_2, \cdots, V_p\}$ of $V$, and an integer $k$;\\
  \textbf{Question:} Is there a selection $S$ of $\mathcal{V}$ such that the chromatic number of the induced graph $G[S]$ is at most $k$, i.e., $\chi(G[S]) \le k$?

  \noindent\rule{\linewidth}{0.3mm}

  We also introduce two known hard problems here, which will be used to prove the hardness results of our problems.

  A \emph{Conjunctive Normal Formula} (CNF) $\phi$ is a conjunction of $m$ given \emph{clauses} $C_1, C_2, \cdots, C_m$ on $n$ boolean variables, where each clause $C_i$ is a disjunction of literals or a single literal and a literal is either a variable or the negation of a variable.
  A literal $x_i$ and its negation $\overline{x}_i$ are called a pair of \emph{contrary literals}.
  A \emph{truth assignment} to $\phi$ is an assignment of the $n$ variables such that every clause in $\phi$ is true.
  The \textsc{$k$-Satisfiability Problem} is defined as follows:

  \noindent\rule{\linewidth}{0.3mm}
  \textbf{The \textsc{$k$-Satisfiability Problem}($k$-SAT)}\\
  \textbf{Input:} A CNF $\phi$  of $m$ given \emph{clauses} $C_1, C_2, \cdots, C_m$ on $n$ boolean variables $x_1, x_2,\cdots,x_n$,
  where each clause contains at most $k$ literals. \\
  \textbf{Question:} Is there a truth assignment to $\phi$?

  \noindent\rule{\linewidth}{0.3mm}

  The $k$-SAT is NP-Complete for each fixed integer $k\ge 3$~\cite{karp1972reducibility}, but polynomially solvable for $k= 1$ or $2$~\cite{krom1967decision}. These results will be used in the proofs of our Theorems \ref{thm_q2_k1}, \ref{thm_q2_k2} and \ref{thm_q3_k1}.

  The \textsc{Independent Set Problem} is another famous problem, which is defined as follows:

  \noindent\rule{\linewidth}{0.3mm}
  \textbf{The \textsc{Independent Set Problem}}\\
  \textbf{Input:} a graph $G=(V,E)$, an integer $k$;\\
  \textbf{Question:} Is there an independent set of size at least $k$ in $G$?

  \noindent\rule{\linewidth}{0.3mm}

  The \textsc{Independent Set Problem} is polynomially solvable when $k$ is a constant and NP-Complete when $k$
  is part of the input~\cite{karp1972reducibility}. Downey and Fellows~\cite{downey1995fixed} further showed that
  the \textsc{Independent Set Problem} is W[1]-hard when taking $k$ as the parameter.
  This W[1]-hardness result implies that the \textsc{Independent Set Problem} will not allow an algorithm with running time
  $f(k) poly(n)$ for any computable function $f(k)$ and polynomial function on the input size $poly(n)$ under the assumption $FPT\neq W[1]$. For more background about parameterized complexity, readers are referred to the monograph~\cite{downey1995fixed}.
  The hardness results of the \textsc{Independent Set Problem} will be used to prove the hardness of PCP,
  say Theorem \ref{p_w1}.

  \section{Complexity of PCP}
  In general, PCP is known to be NP-hard since it contains the well known NP-hard problem, the \textsc{Vertex Coloring Problem}
  as a special case, where each part contains exactly one vertex.
  In this paper, we will consider the complexity of PCP with  respect to the following three parameters:
  \begin{itemize}
    \item the number of colors, $k$;
    \item the number of parts in the partition, $p$;
    \item the maximum cardinality among all parts in the partition, $q$.
  \end{itemize}

  We will show that PCP is polynomially solvable only when some of the three parameters are small constants.
  First of all, it is trivially to see that PCP is in NP for any setting of the three parameters.
  Given an assignment of colors to a subset of vertices, we can easily check whether it is a selection and a feasible $k$-coloring in polynomial time.

  \begin{theorem}\label{order}
  PCP is in NP.
  \end{theorem}

  \subsection{Parameters $q$ and $k$}
  We now discuss the NP-Hardness of PCP
  with different constant values of $q$ and $k$.
  As mentioned above, when $q=1$, the problem is equal to VCP, which is NP-Complete for each constant $k\ge 3$ and polynomially solvable for each $k\le 2$~\cite{karp1972reducibility}. Therefore, we have the following conclusion.

  \begin{theorem}\label{thm_q1_k1_k2_k3}
  When $q = 1$, PCP is polynomially solvable for each constant $1\leq k \le 2$ and NP-Complete for each constant $k\ge 3$.
  \end{theorem}




  Next, we consider the cases where $q\geq 2$.

  \begin{theorem}\label{thm_q2_k1}
  When $q=2$ and $k=1$, PCP is polynomially solvable.
  \end{theorem}
  \begin{proof}
  We show that the case that $k=1$ and $q=2$ can be polynomially reduced to the polynomially solvable problem 2-SAT~\cite{krom1967decision}.

  For an instance of PCP with $q=2$ and $k=1$, a graph $G=(V,E)$ and a $p$-partition $\mathcal{V}$ of $V$ where each part has at most 2 vertices, we construct
  an instance $\phi$ of 2-SAT on $p$ variables.

  For each part $V_j$ in $\mathcal{V}$, we associate it with a variable $x_j$. Then we have $p$ variables in $\phi$.
  Furthermore, we associate each vertex in $G$ with a literal (either a variable $x$ or its negative $\overline{x}$):
  for each part of size 2, say $V_j=\{u_{j1}, u_{j2}\}$, we associate vertex $u_{j1}$ with literal $\ell_{j1}=x_j$ and associate vertex $u_{j2}$ with literal $\ell_{j2}=\overline{x}_j$;
  for each part of size 1, say $V_j=\{u_{j1}\}$, we associate vertex $u_{j1}$ with literal $\ell_{j1}=x_j$.
  Next, we construct clauses. We will have $p$ \emph{vertex clauses} and $|E|$ \emph{edge clauses}.
  For each part $V_j$ of size 1, we construct a vertex clause $\ell_{j1}$ containing exactly one literal;
  for each part $V_j$ of size 2, we construct a vertex clause $\ell_{j1}\vee \ell_{j2}$ of size 2. We can see that the second kind of vertex clause will always be true since $\ell_{j2}=
  \overline{\ell_{j1}}$. However, we keep them for the purpose of presentation.
  For each edge $(u,v)\in E$, we construct an edge clause $\overline{\ell_u} \vee \overline{\ell_v}$ of size 2, where $u$ is associated with the literal $\ell_u$ and
  $v$ is associated with the literal $\ell_v$. Thus, we have $|E|$ edge clauses.

  We prove that there is a selection of $\mathcal{V}$ which is 1-colorable if and only if $\phi$ is satisfied.

  \textbf{The ``$\Rightarrow$'' part:}
  Assume there is a selection $S$ of $\mathcal{V}$ which is 1-colorable. Then $S$ will form an independent set. For each vertex $v$ in $S$, we let its associated literal $\ell_v$ be 1.
  For any variable left without assigning a value, we simply let it be 1.
  We claim that this is a truth assignment to $\phi$. Since each vertex is associated with a different literal, we know the above assignment of letting the literals associated to vertices in $S$ is feasible. First of all, we know each edge clause is satisfied since each part contains at least one vertex in $S$ and then each edge clause contains at least one literal with value 1. For an edge clause $\overline{\ell_u} \vee \overline{\ell_v}$ corresponding to the edge $(u,v)$, if it is not satisfied, then $\ell_u=\ell_v=1$ and thus both of $u$ and $v$ are in $S$, which is a contradiction to the fact that $S$ is an independent set. So all edge clauses are satisfied and $\phi$ is satisfied.

  \textbf{The ``$\Leftarrow$'' part:}
  Assume that $\phi$ is satisfied. For a truth assignment $A$ of $\phi$, we select a vertex $v$ into the selection $S$ if and only if its associated literal is assigned 1.
  Then the set $S$ is a 1-colorable selection. The reason is as follows.
  Each vertex clause can have at most one literal of value 1 in $A$. So each part has one vertex being selected into $S$.
  For any two vertices $u,v \in S$, if there is an edge between them, then there is an edge clause $\overline{\ell_u} \vee \overline{\ell_v}$. Since $\overline{\ell_u} \vee \overline{\ell_v}$ should be 1 in $A$, we know that at least one of $\ell_u$ and $\ell_v$ is 0 and then at least one of $u$ and $v$ is not in $S$, a contradiction.
  So there is no edge between any two vertices in $S$ and such $S$ is an independent set.
  \hfill \qed
  \end{proof}

  \begin{theorem}\label{thm_q2_k2}
  When $q= 2$ and $k= 2$, PCP is NP-Complete.
  \end{theorem}
  \begin{proof}
  Theorem~\ref{order} shows that the problem is in NP.
  For the NP-Hardness, we give a reduction from the known NP-Complete problem 3-SAT to PCP with $q=2$ and $k=2$.

  Let $\phi$ be a 3-SAT formula of
  $m$ clauses $C_1, C_2, \cdots, C_m$ on $n$ boolean variables $x_1, x_2,\cdots,x_n$, where we can assume that $|C_t| = 3$ holds for each clause $C_t$.
  We construct an instance of PCP. The graph $G=(V,E)$ contains $|V| = 9m+2$ vertices. In the $p$-partition $\mathcal{V}$,
  each part has most $q=2$ vertices and $p=6m+2$. We will show that $G$ has a selection $S$ of $\mathcal{V}$ such that
  the chromatic number of $G[S]$ is at most $k=2$ if and only if $\phi$ is satisfiable.

  The graph $G$ is constructed in the following way.

  First, we  introduce two vertices denoted by $g$ and $r$.
  Then, for each clause $C_t= (a^t_1\vee a^t_2\vee a^t_3)$ ($t\in\{1,\cdots,m\}$) in $\phi$, we introduce $9$ vertices that are divided into three parts of three vertices, called the \emph{literal layer}, the \emph{middle layer} and the \emph{conflict layer}.
  The three vertices in the literal layer are denoted by $l^t_1, l^t_2, l^t_3$,
  the three vertices in the middle layer are denoted by $m^t_1, m^t_2, m^t_3$, and
  the three vertices in the conflict layer are denoted by $c^t_1, c^t_2, c^t_3$.
  For $i\in {1,2,3}$ the four vertices $l^t_i$, $m^t_i$ and $c^t_i$ are associated with the literal $a^t_i$.
  In total, the graph has $9m+2$ vertices.

  For edges in the graph $G$, we first add an edge between $g$ and $r$.
  Then, for each clause $C_t$, we introduce $9$ edges as follows.
        \begin{itemize}
          \item Connect $l^t_i$ to $m^t_i$ for each each $i\in \{1,2,3\}$ (3 edges);
          \item Connect each vertex $m^t_i$ ($i\in \{1,2,3\}$) in the middle layer to $g$ (3 edges);
          \item Connect each pair of vertices in the conflict layer to form a triangle (3 edges).

        \end{itemize}
  Last, for each pair of contrary literals ${a^{t_1}_i}$ and
   $a^{t_2}_j$ ($i,j\in\{1,2,3\}, t_1,t_2 \in\{1,\cdots,m\}$) in $\phi$,
  add an edge between the two vertices associated with ${a^{t_1}_i}$ and $a^{t_2}_j$ in the literal layers.

  In terms of the $p$-partition $\mathcal{V}$, we will have $p=6m+2$ parts, each of which contains at most $q=2$ vertices.
  Vertices $g$ and $r$ form two separated parts containing one vertex, $\{g\}$ and $\{r\}$.
  For each clause $C_t$, the 9 vertices associated with it will be divided into 6 parts: $\{l^t_i\}$ and $\{m^t_i, c^t_i\}$ for $i=1,2,3$.

  \begin{figure}
  \centering
  \includegraphics[width=4.0in]{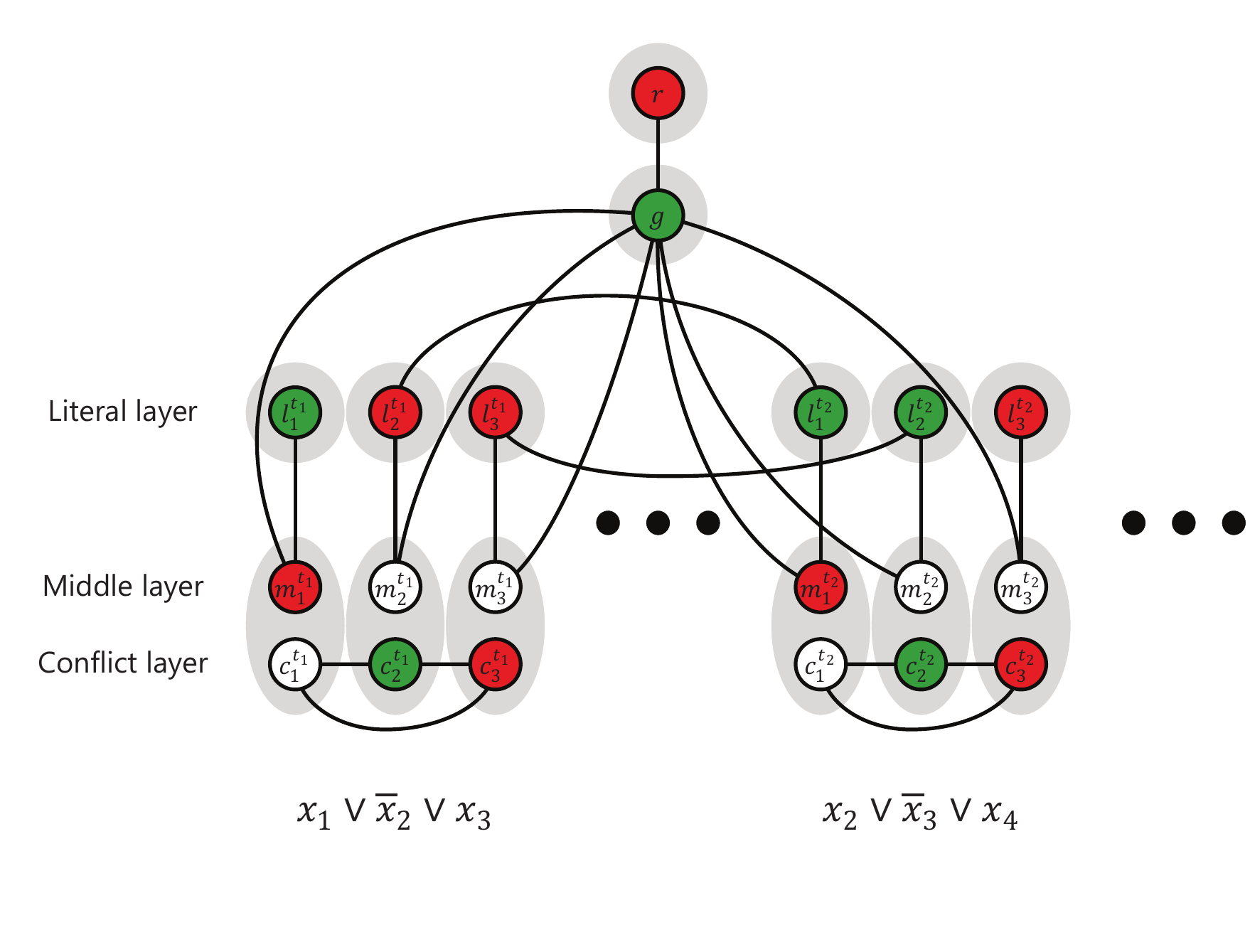}
  \captionsetup{justification=centering}
  \caption{An example of the construction with two clauses $C_{t_1} = x_1\vee \overline{x}_2\vee x_3$ and $C_{t_2}=x_2\vee \overline{x}_3\vee x_4$, where $x_1$ and $x_2$ have value true, and $x_3$ and $x_4$ have value false. In the figure. a grey shadow represents a part of the partition.}\label{fig1}
  \end{figure}

  An illustration of the construction is shown in Figure~\ref{fig1}.
  We now show that CNF $\phi$ is satisfiable if and only if there is a selection $S$ of $\mathcal{V}$ such that $G[S]$ is $2$-colorable.

  \textbf{The ``$\Leftarrow$'' part:} Assume that $S$ is a selection of $\mathcal{V}$ such that $G[S]$ is $2$-colorable
  and let $c:S \rightarrow \{green, red\}$ be a 2-coloring of $G[S]$.
  Since there is an edge between $g$ and $r$, we know that $c(g)\neq c(r)$. With loss of generality, we assume that
  $c(g)= green$ and $c(r)= red$.

  Each vertex in the literal layers must be in the selection $S$ since it is in a part of size 1.
  We claim that
  \begin{property}\label{lemma_one true}
   For the three vertices $l^t_1, l^t_2$ and  $l^t_3$ ($t\in \{1,\cdots, m\}$) in a literal layer associated to the clause $C_t$, at least one of them is assigned $green$ in the 2-coloring $c$.
  \end{property}
  Assume to the contrary that $c(l^{t}_1)= c(l^{t}_2)= c(l^{t}_3) =red$. For this case, none of $m_1^{t}, m_2^{t}$ $m_3^{t}$ can be assigned to
  either $red$ or $green$ and then none of them is in $S$. Such all of $c_1^{t}, c_2^{t}$ and $c_3^{t}$
  are in $S$. Note that $c_1^{t}, c_2^{t}$ and $c_3^{t}$ form a triangle. It is impossible to color them by using only 2 colors,
  a contraction. So Property~\ref{lemma_one true} holds.

  For each clause $C_t$ ($t\in \{1,\ldots, m\}$), we select an arbitrary vertex $l_i^t$ in the literal layer with $green$ color in the 2-coloring $c$ and assign the corresponding literal $a_i^t$ value 1. After doing this, if there are still variables without assigned the value, arbitrarily assign 1 or 0 to it. We claim the above assignment of the variables is a truth assignment to $\phi$.
  Note that there is an edge between any pair of contrary literals in the graph $G$. So it is impossible that two
  contrary literals are assigned $green$ in $c$. Therefore, the above assignment is a feasible assignment.
  Furthermore, by Property~\ref{lemma_one true}, we know that each clause will have at least one literal assigned value 1.
  Therefore, it is a truth assignment to $\phi$.

  \textbf{The ``$\Rightarrow$'' part:}
  Assume there is a truth assignment of $\phi$. We show that there is a selection $S$ of $\mathcal{V}$ and
  a 2-coloring $c$ of $G[S]$.
  First of all, vertices $g$ and $r$ are selected into $S$. We let $c(g)=green$ and $c(r)=red$.
  Second, all vertices in the literal layers are selected into $S$, and a vertex $l^t_i$ in the literal layers
  is assigned to color $green$ (resp., $red$) in $c$ if the corresponding literal $a_i^t$ has value 1 (resp., 0).
  Third, for the parts containing a vertex in the middle layer and a vertex in the conflict layer,
  we select the vertex $m^t_i$ in the middle layer into $S$ and assign color $red$ to it if the corresponding vertex $l^t_i$
  in the literal layer is assigned to color $green$; we select the vertex $c^t_i$ in the conflict layer into $S$
  if the corresponding vertex $l^t_i$
  in the literal layer is assigned to color $red$. For the color of $c^t_i$, if no neighbor of $c^t_i$ in $G[S]$ has been assigned
  a color, we assign color $green$ to it, and otherwise, we assign color $red$ to it.
  We argue that above coloring is a feasible 2-coloring of $G[S]$.
  After the second step, no adjacent vertices are assigned the same color because there are only edges between pairs of contrary literals. In the third step, we can assign a $red$ color to vertices $m^t_i$ in the middle layer because they are only adjacent to vertices $T$ and $l^t_i$, both of which are assigned $green$. For the vertices $c^t_i$ in the conflict layer, each of them is adjacent to at
  most one vertex in $G[S]$ because at least one vertex in the literal layer of each clause $C_t$ is assigned $green$ by Property~\ref{lemma_one true} and then at most two vertices in the conflict layer of $C_t$ can be selected into $S$.
  So our way to color vertices in conflict layer is correct.
    \hfill \qed   \end{proof}

  Readers are refer to Figure~\ref{fig1} for an illustration of the reduction of the case $q=2$ and $k=2$, where $\phi$ contains at least two
  clauses $C_{t_1}=x_1\vee \overline{x}_2\vee x_3$ and $C_{t_2}=x_2\vee \overline{x}_3\vee x_4$. When $x_1$ and $x_2$ have value true, $x_3$ and $x_4$ have value false, then we assign $l_1^{t_1}$, $l_1^{t_2}$, $l_2^{t_2}$, $c_2^{t_1}$ and $c_2^{t_2}$ $green$ and we assign $l_2^{t_1}$, $l_3^{t_1}$, $l_3^{t_2}$, $m_1^{t_1}$, $m_1^{t_2}$, $c_3^{t_1}$ and $c_3^{t_2}$ $red$.


  Now let us extend the below theorem to the cases $q\ge 2$ and $k\ge 2$.

  \begin{corollary}\label{coro_q2_k2}
  For each constant $q\ge 2$ and for each constant $k\ge 2$, PCP is NP-Complete.
  \end{corollary}

  \begin{proof}
  For any fixed constants $q'\ge 2$ and $k'\ge 2$, we show that the case that $q=2$ and $k=2$ can be polynomially reduced to the case that $q=q'$ and $k=k'$.

  For an instance of PCP with $q=2$ and $k=2$, a graph $G=(V,E)$ and a $p$-partition $\mathcal{V}$ of $V$, we construct another instance of PCP with $q=q'$ and $k=k'$: a graph $G'=(V',E')$ and a $p'$-partition $\mathcal{V}'$ of $V'$.

  The graph $G'$ contains a copy of $G=(V,E)$, a complete graph $K_{k'-2}=(V_0=\{w_1,w_2,\cdots,w_{k'-2}\}, E_0)$ on $k'-2$ vertices and $q'$ vertices $u_1, u_2,\cdots,u_{q'}$.
  To finish the construction of $G'$, we connect all the vertices in $V$ to all the vertices in $K_{k'-2}$.

  The $p'$-partition $\mathcal{V}'$ is given as follows: $$\mathcal{V'}=\mathcal{V}\cup\{\{w_1\},\{w_2\},\cdots,\{w_{k'-2}\},\{u_1,u_2,\cdots,u_{q'}\}\}.$$
  So $p'=p+k'-1$.

  We now show that there is a selection $S$ of $G$ such that $G[S]$ is $2$-colorable if and only if there is a selection $S'$ of $G'$ such that $G'[S']$ is $k'$-colorable.

  \textbf{The ``$\Leftarrow$'' part:}
  Given a selection $S$ that is $2$-colorable in $G$, we construct $S'$ by copying $S$ and adding all the vertices in $K_{k'-2}$ and an arbitrary vertex $u\in \{u_1,u_2,\cdots,u_{q'}\}\}$.
  We can see that $G'[S']$ is $k'$-colorable: since $S$ can be assigned by $2$ colors, we can assign vertices in $K_{k'-2}$ by another $k'-2$ colors and assign $u$ an arbitrary color.

  \textbf{The ``$\Rightarrow$'' part:}
  Given a selection $S'$ that is $k'$-colorable in $G'$, we let $S=S'\cap V$ and show that $S$ is $2$-colorable in $G$.
  We can see that $S'$ must contain all the vertices in $K_{k'-2}$.
  Since the vertices in $K_{k'-2}$ are assigned by $k'-2$ different colors, the vertices in $S$ are assigned by another $2$ colors.
    \hfill \qed   \end{proof}

  \begin{theorem}\label{thm_q3_k1}
  When $k=1$, PCP is NP-Complete for each constant $q\geq3$.
  \end{theorem}

  \begin{proof}
  Theorem~\ref{order} shows that the problem is in NP. We only need to prove the NP-hardness.
  We give a reduction from $q$-SAT to PCP with $k=1$.
  Note that the $q$-SAT is NP-Complete when $q\ge 3$.

  Let $\phi$ be a $q$-SAT formula of
  $m$ clauses $C_1, C_2, \cdots, C_m$ on $n$ boolean variables $x_1, x_2,\cdots,x_n$, where $|C_t| \leq q$ holds for each clause $C_t$.
  We construct a PCP instance of a graph $G$ and a partition $\mathcal{V}$.
  The graph $G$ is constructed as follows: for each literal in a clause, we introduce a vertex associated with the literal. For each pair of contrary literals, we add an edge between the two associated vertices.
  The partition $\mathcal{V}$ is exactly obtained according to the clauses of $\phi$, i.e., $\mathcal{V}$ contains $m$ parts and
  each part $V_t\in \mathcal{V}$ contains the vertices associated to literals in the clause $C_t$.
  Thus, each part has size at most $q$.
  We now claim that there is a selection $S$ of $\mathcal{V}$ such that the induced graph $G[S]$ is $1$-colorable if and only if the formula $\phi$ is satisfiable.

  \textbf{The ``$\Rightarrow$'' part:}
  Given a selection $S$ that is $1$-colorable, we assign value 1 to the literals associated with the vertices in $S$.
  Note that $S$ is an independent set. So no two contrary literals will be assigned the same value and the above assignment is feasible.
  After this, if there are any variables left without assigning a value, we assign arbitrary values to them.
  It is safe to say that the above assignment is a truth assignment to $\phi$ as each clause has at least one literal being value 1.

  \textbf{The ``$\Leftarrow$'' part:}
  If $\phi$ is satisfiable, we can construct a selection $S$ of the partition $\mathcal{V}$ such that $G[S]$ is $1$-colorable.
  Let $A$ be a truth assignment to $\phi$.
  For any part $V_t$ in $\mathcal{V}$, the associated clause $C_t$ must contain at least one literal with value 1 since in $A$ since $A$ is a truth assignment.
  We arbitrary select a vertex $v\in V_i$ such that  the literal associated to $v$ has value 1.
  Since no two contrary literals are assigned value 1 simultaneously, we know that the above selection gets an independent set.
    \hfill \qed   \end{proof}

  \subsection{Parameter $p$}
  Next, we consider the parameter $p$. It is easy to see that PCP is polynomially solvable when $p$ is a constant.
  A simple brute-force algorithm runs in polynomial time: by enumerating all vertices in each part to search the selection $S$ we will get at most $\prod_{i=1}^p |V_i|$ candidates for $S$ in $\prod_{i=1}^p |V_i|\leq n^p$ time; for each candidate we can check whether it is $k$-colorable in $O(p^k)$ time, where $k\leq p$. When $p$ is a constant, the algorithm runs in polynomial time.

  \begin{theorem}\label{thm_pconstant}
  When $p$ is a constant, PCP is polynomially solvable.
  \end{theorem}

  Since PCP is NP-hard when $p$ is part of the input and polynomially solvable when $p$ is a constant, it is reasonable to consider whether PCP is fixed-parameter tractable by taking parameter $p$. We have the following negative result.

  \begin{theorem}\label{p_w1}
  Taking $p$ as the parameter, PCP is W[1]-hard even for each fixed $k\geq 1$.
  \end{theorem}
 
  \begin{proof}
  We show that for each fixed $k\geq 1$, there is an FPT reduction from the \textsc{$k_I$-Independent Set Problem} to PCP.

  Given an instance of the \textsc{$k_I$-Independent Set Problem}, a graph  $G=(V,E)$ and an integer $k_I$, we construct an instance of PCP.
  The graph $G'$ contains a complete graph  $K_{k-1}=(V_0=\{w_1, w_2, \cdots, w_{k-1}\},E_0)$ on $k-1$ vertices and $k_I$ copies of $G$, namely, $G_1=(V_1, E_1), G_2=(V_2, E_2), \cdots, G_{k_I}=(V_{k_I}, E_{k_I})$,
  where $K_{k-1}$ is regarded as empty when $k=1$.
  We still need to add some edges to finish the construction of the graph $G'$:
  for two vertices $v_i\in V_i$ and $v_j\in V_j$ from two different copies of $G$ ($i,j \in \{1,2,\dots, k-1\}$),
  if $v_i$ and $v_j$ are corresponding the same vertex or a pair of adjacent vertices in $G$, then we add an edge between them;
  for any vertex pair $u\in V_0$ and  $v\in V_1\cup\cdots\cup V_{k_I}$, we add an edge between them, and thus it becomes a complete bipartite graph between $V_0$ and  $V_1\cup\cdots\cup V_{k_I}$.
  The $p$-partition is given as follows: $\mathcal{V}=\{V_1, V_2, \cdots, V_{k_I}, \{w_1\}, \{w_2\},\cdots, \{w_{k-1}\}\}$. So $p=k_I + k - 1$. We claim that $G$ has an independent set of size $k_I$ if and only if  $G'$ has a selection $S$ of $\mathcal{V}$ such that $G'[S]$ is $k$-colorable.

  \textbf{The ``$\Rightarrow$'' part:}
  Assume that there is an independent set $I$ of $G$ with size $|I|=k_I$. Let $I = \{i_1, \cdots, i_{k_I}\}$.
  We find a solution to PCP instance as follows. From each of the $k_I$ copies of $G$,
  we select one vertex corresponding to a different vertex in the independent set $I$.
  So the $k_I$ vertices selected form the first $k_I$ parts in $\mathcal{V}$ will form an independent set. We color all the $k_I$ vertices with one color.
   All the $k-1$ vertices in $V_0$ will be also selected into $S'$ since each of them is in a part of a single vertex. We color these vertices with the other $k-1$ colors.
   Thus, $G'[S]$ is $k$-colorable.

  \textbf{The ``$\Leftarrow$'' part:}
  Assume that there is a selection $S$ of $\mathcal{V}$ such that $G'[S]$ is $k$-colorable.
  We show that $G$ has an independent set of size $k_I$.
  Since each vertex in $V_{0}$ is in a part of a single vertex in the partition $\mathcal{V}$, all the vertices in $V_{0}$ are in $S$.
  Let $I=S\setminus V_0$.
  Note that $I$ and $V_0$ form a complete bipartite graph in $G'[S]$, and the chromatic number of $G_0=(V_0,E_0)$ is $k-1$.
  If $G'[S]$ is $k$-colorable, then all vertices in $I$ should receive the same color in a $k$-coloring.
  So we know that $I$ is an independent set in $G'$, which also implies that the vertices in the original graph $G$ corresponding $I$
  form an independent set of size $k_I$ in $G$.
    \hfill \qed   \end{proof}

  On the other hand, it is easy to see that PCP is FPT when both of $p$ and $q$ are taking as the parameters.
  \begin{theorem}\label{p_FPT}
  Taking $p$ and $q$ as the parameters, PCP is fixed-parameter tractable.
  \end{theorem}
  \begin{proof}
  In fact, a simple brute-force algorithm is FPT. If $k>p$, the problem has no solution. Next, we assume that $k\leq p$.
  We enumerate all possible selections, the number of which is at most $q^p$. For each candidate selection, there are at most $k^p\leq p^p$ different ways to color them.
  To check whether a color is feasible can be done in linear time. So the algorithm runs in $O((pq)^p (|V|+|E|))$ time.
    \hfill \qed   \end{proof}

  \section{An Exact Algorithm for PCP}
  In this section, we consider fast exact algorithms for PCP.
  It is known that VCP can be solved in $O^*(2^n)$\footnote{The notation $O^*$ is a modified big-O notation that suppresses all polynomially bounded factors.} time by using subset convolution~\cite{bjorklund2007fourier}, while traditional dynamic programming algorithms can only lead to running time of  $O^*(3^n)$.
  By using the  $O^*(2^n)$-time algorithm for VCP, we can get a simple $O^*((\frac{2n}{p})^p)$-time algorithm for PCP: we enumerate all candidates of the selection and check whether they
  are $k$-colorable. The number of candidates of the selection is $\prod_{i=1}^{p}|V_i|$, which is at most $(\frac{n}{p})^p$ by the AM-GM inequality~\cite{cauchy1821cours}. Each candidate is a part of $p$ vertices and we use
  the $O^*(2^p)$-time algorithm to check whether it is $k$-colorable. So in total, the algorithm runs in $O^*((\frac{2n}{p})^p)$ time.
  Next, we use the subset convolution technique to improve the running time bounded to $O((\frac{n+p}{p})^pn\log k)$. Note that when $p=n$, the problem becomes VCP and the running time bound reaches the best known bound $O^*(2^n)$ for VCP.

  \begin{definition}\label{semi}
  Given a $p$-partition $\mathcal{V}=\{V_1,\cdots,V_p\}$, a vertex subset $S\subseteq V$ is called a \emph{semi-selection} if
  it holds that $|S\cap V_i|\le 1$  for all $1\le i \le p$.
  \end{definition}

   It is easy to see that a semi-selection $S$ with size $|S|=p$ is a selection.
   We use $\mathcal{S}$ denotes the set of all semi-selections corresponding to a $p$-partition $\mathcal{V}$.
  We have that
  \begin{equation} \label{e1}
  |\mathcal{S}|\le (\frac{n+p}{p})^p.
  \end{equation}
  Each semi-selection $S$ has at most one vertex in each part $V_i$. There are $|V_i| + 1$ possibilities for $S\cap V_i$. Such the number of
    semi-selections is $\mathcal{|S|}=\prod_{i=1}^p(|V_i| + 1)$.
  Since $\sum_{i=1}^p(|V_i|+1)=n+p$, by the AM-GM inequality, we have
  $\mathcal{|S|}\le (\frac{n+p}{p})^p$.

  The set $\mathcal{S}$ is a hereditary family, that is to say, for any semi-selection $S\in \mathcal{S}$, all the subsets of $S$ are also semi-selections.
  We will use the following lemma to design our algorithm.

  \begin{theorem}\label{ext_semi}
  Let $S\in \mathcal{S}$ be a semi-selection. Then $G[S]$ is $k$-colorable if there is a subset $T\subseteq S$ such that $G[T]$ is 1-colorable and $G[S\setminus T]$ is $(k-1)$-colorable.
  \end{theorem}
  \begin{proof}
  Since $\mathcal{S}$ is a hereditary family, we know that each subset of $S$ is also a semi-selection.
  For a $k$-coloring of $G[S]$, the set of vertices with the same color is a satisfied subset $T$.
    \hfill \qed   \end{proof}

  Before introducing our algorithm, we give the definition of the subset convolution first.
  \begin{definition}\label{subCon}
  Let $\mathcal{S}$ be a hereditary family on a set contains $n$ elements and $g,h$ be two integer functions on $\mathcal{S}$, i.e., $g,h:\mathcal{S} \rightarrow \mathbb{Z}$. The subset convolution of $g$ and $h$, denoted by $g*h$, is a function assigning to any $S\in \mathcal{S}$ an integer
  $$(g*h)(S)=\sum_{T\subseteq S} g(T)\cdot h(S\setminus T).$$
  \end{definition}
  The subset convolution can be computed in time $O(n|\mathcal{S}|)$~\cite{bjorklund2007fourier}.

  We show how to use the subset convolution to solve PCP.
  Let $f$ be an indicator function on  semi-selections $f:\mathcal{S} \rightarrow \{0,1\}$. For any semi-selection $S$, $f(S)=1$ if $G[S]$ is 1-colorable, and $f(S)=0$ otherwise.
  Define $f^{*k}:\mathcal{S} \rightarrow \mathbb{Z}$ as follows:
  $$f^{*k}=\underbrace{f*f*\cdots*f}_{k \text{ times}}.$$
  We have the following theorem
  \begin{theorem}\label{colorCon}
  For any semi-selection $S\in \mathcal{S}$, the graph $G[S]$ is $k$-colorable if and only if $f^{*k}(S) > 0$.
  \end{theorem}

  \begin{proof}
    Since $f^{*k}=f*f^{*k-1}$, we can write $f^{*k}(S)$ in the following way
  $$\begin{aligned}
  f^{*k}(S)&=\sum_{S_1\subseteq S} f(S_1)\cdot f^{*k-1}(S\setminus S_1)\\
  &=\sum_{S_1\subseteq S}\sum_{S_2\subseteq (S\setminus S_1)}f(S_1)\cdot f(S_2)\cdot f^{*k-2}(S\setminus (S_1\cup S_2)) \\
  &\vdots\\
  &=\sum_{S_1,S_2,\cdots,S_k}\prod_{i=1}^kf(S_i),
  \end{aligned}$$
  where $S_1,S_2,\cdots,S_k$ form a partition of the semi-selection $S$.

  If $f^{*k} > 0$, then there must be a partition $\{S_1,S_2,\cdots,S_k\}$ of $S$ such that $\prod_{i=1}^kf(S_i) = 1$ which induces that $f(S_i)=1$ for each $S_i$, as a consequence $G[S]$ is $k$-colorable.

  In the opposite direction, $G[S]$ is $k$-colorable, so we can divide $S$ into $k$ $1$-colorable semi-selections $S_1,S_2,\cdots,S_k$ where $\prod_{i=1}^k f(S_i)=1$. Thus $f^*(S)>0$.
    \hfill \qed   \end{proof}

  By this theorem, to check whether there is a $k$-colorable selection, we only need to check whether there is a semi-selection $S$ such that $|S|=n$ and  $f^{*k}(S)>0$.
  The detailed steps of the algorithm is given below.
  \begin{algorithm}[H]
  \caption{An exact algorithm for PCP.}\label{alg1}
  \begin{algorithmic}[1]
  \REQUIRE A simple undirected graph $G=(V,E)$, a $p$-partition $\mathcal{V}$ and an integer $k$.\\
  \ENSURE `yes' or `no' to indicate wether there exits a selection $S$ such that $G[S]$ is $k$-colorable.\\
  \STATE Enumerate all semi-selections and store them in $\mathcal{S}$;
  \STATE Check all the semi-selections $S$ whether they are 1-colorable, and let $f(S)=1$ if $S$ is 1-colorable and $f(S)=0$ otherwise;
  \STATE Calculate the subset convolution $f^{*k}$;
  \STATE If there is a semi-selection $S$ such that $|S|=p$ and $f^{*k}(S) > 0$, stop and return `yes';
  \STATE return `no'.
  \end{algorithmic}
  \end{algorithm}

  \begin{theorem}\label{theorem time}
  Algorithm~\ref{alg1} solves PCP in $O((\frac{n+p}{p})^pn\log k)$ time.
  \end{theorem}
  \begin{proof}
  The correctness of Algorithm~\ref{alg1} follows from Theorem~\ref{colorCon} directly. Next, we analyze the running time bound.
  The size of each semi-selection is at most $p$. So Steps $1$, $2$ and $4$ in Algorithm $1$ can be implemented in time $O(|\mathcal{S}|p)$.
  The fast subset convolution algorithm computes $f*f$ in time $O(|\mathcal{S}|n)$~\cite{bjorklund2007fourier}.
  In Step 3, we need to compute $f^{*k}$. Instead of computing $k$ subset convolutions,  we apply a doubling trick, that is, to compute $f^{*k}$ by using the following recurrence relation
  $$f^{*k}=\begin{cases}
  f^{*\lfloor\frac{k}{2}\rfloor}*f^{*\lfloor\frac{k}{2}\rfloor}& k>1 \text{ and } k \text{ is even;}\\
  f^{*\lfloor\frac{k}{2}\rfloor}*f^{*\lfloor\frac{k}{2}\rfloor}*f& k>1 \text{ and } k \text{ is odd;}\\
  f& k=1\text{.}
  \end{cases}$$
  Thus, we only need to compute $\log k$ subset convolutions. Step~3 uses $O(|\mathcal{S}|p\log k)$ time.

  By (\ref{e1}), we have that $|\mathcal{S}|\le (\frac{n+p}{p})^p$. The algorithm runs in $O((\frac{n+p}{p})^pn\log k)$ time.
    \hfill \qed   \end{proof}

  We have proved that  PCP with parameter $p$ is W[1]-hard. It is unlikely to remove $n$ from the exponential part of the running time.
  Furthermore, when $p=n$, the problem becomes VCP and the running time bound $O^*(2^n)$, which is the best-known result for VCP.

  \section{Concluding Remarks}
  In this paper, we have analyzed the computational complexity of the \textsc{Partition Coloring Problem}.
  By reducing from the \textsc{Independent Set Problem}, the \textsc{Vertex Coloring Problem} and the \textsc{$k$-SAT problem}, we show different NP-hardness results of the \textsc{Partition Coloring Problem} with respect to different constant settings of three parameters: the number of colors, the number of parts of the partition, and the maximum size of each part of the partition. We also design polynomial-time algorithms for the remaining cases. It would be interesting to look into the complexity status of the problem in subgraph classes with different settings on the three parameters.

   \section*{Acknowledgement}
   This work was supported by the National Natural Science Foundation of China,
under grants 61972070 and 61772115.

\end{document}